\documentclass[english]{article}
\usepackage{amsfonts,amsmath,amsthm,amscd,amssymb,latexsym}
\usepackage{tikz}

\theoremstyle{plain}
\newtheorem{thm}{Theorem}[section]
\newtheorem{cor}[thm]{Corollary}

\newtheorem{prop}[thm]{Proposition}

\theoremstyle{definition}
\theoremstyle{remark}
\numberwithin{equation}{section}
\newcommand{\keywords}{\textbf{Key words and phrases: }\medskip}
\newcommand{\subjclass}{\textbf{Math. Subj. Clas.: }\medskip}

\newcommand{\bigzero}{\mbox{\normalfont\large 0}}
\begin{document}
\title{\textbf{2D  discrete Hodge-Dirac operator on the torus} }
\author{\textbf{Volodymyr Sushch} \\
{ \em Koszalin University of Technology} \\
 {\em Sniadeckich 2, 75-453 Koszalin, Poland} \\
 { \em volodymyr.sushch@tu.koszalin.pl} }

\date{}
\maketitle
\begin{abstract}
We discuss a discretisation of the de Rham-Hodge theory in the two-dimensional case based on a discrete exterior calculus framework.
We present discrete analogues  of the Hodge-Dirac and Laplace operators in which  key geometric aspects of the continuum counterpart are captured.
We provide and prove a discrete version of the Hodge decomposition theorem. Special attention has been paid to discrete models on  a combinatorial torus. In this particular case, we also define and calculate the cohomology groups.
\end{abstract}

\keywords{discrete exterior calculus, discrete operators, Hodge-Dirac operator, discrete Laplacian, Hodge decomposition, combinatorial torus, cohomology groups}

 \subjclass  {39A12, 39A70, 58A14}

 \section{Introduction}
The choice of technique to approximate the solution of partial differential equations depends on a discretisation scheme.
The interest for discrete models which preserve a geometric structure of continuum counterparts has grown in the computing community \cite{Arnold1, Arnold2, Ayoub, Nitschke, Stern}.
It is known that a  geometric discretisation scheme would be ideal if this had the same properties as the continuum. However, there are difficulties usually in definitions of discrete counterparts of the Hodge star and the wedge product between forms. Various approaches with geometric discretisation have been proposed in the literature. See, for example, \cite{Beauce, Becher, Dodziuk, Rabin, SSSA,Teixeira, Watterson, Wilson}.
 Recently, in \cite{CKSV}, it was described a quite general framework of  discrete calculus  based on a new  type of discrete
geometry called script geometry. The proposed approach in this article was introduced by Dezin \cite{Dezin} and later developed in the author's previous papers \cite{S1,S2,S3,S4,S5,S6}.

 Let $\Lambda(\mathbb{R}^2)=\Lambda^0(\mathbb{R}^2)\oplus\Lambda^1(\mathbb{R}^2)\oplus\Lambda^2(\mathbb{R}^2)$ denote the graded vector space of smooth differential forms on $\mathbb{R}^2$, where $\Lambda^r(\mathbb{R}^2)$ denotes the subspace of $r$-forms, $r=0,1,2$. Let
$d: \Lambda^r(\mathbb{R}^2) \rightarrow \Lambda^{r+1}(\mathbb{R}^2)$ be the exterior derivative. The codifferential
$\delta: \Lambda^r(\mathbb{R}^2) \rightarrow \Lambda^{r-1}(\mathbb{R}^2)$ is defined by $\delta=-\ast d \ast$, where $\ast$ is the Hodge star operator such that $\ast: \Lambda^r(\mathbb{R}^2) \rightarrow \Lambda^{2-r}(\mathbb{R}^2)$ and $\ast^2=(-1)^r$.
The operator
\begin{equation}\label{1.1}
d+\delta: \Lambda(\mathbb{R}^2) \rightarrow \Lambda(\mathbb{R}^2)
\end{equation}
is called the Hodge-Dirac operator on $\mathbb{R}^2$.
The Laplacian $\Delta: \Lambda^r(\mathbb{R}^2) \rightarrow \Lambda^r(\mathbb{R}^2)$ is defined by
\begin{equation}\label{1.2}
\Delta=(d+\delta)^2=d\delta+\delta d.
\end{equation}
Our purpose is to develop a satisfactory discrete model of de Rham-Hodge theory on manifolds which are homeomorphic to the torus.
We consider a chain complex as a combinatorial model of $\mathbb{R}^2$. When we add to this discrete analogues of the exterior derivative, the Hodge star operator, and the exterior product acting on cochains, we have all of the basic ingredients for the calculus of discrete counterparts of differential forms. We show that discrete analogues of the operators \eqref{1.1} and \eqref{1.2} have properties like those in the continual case. We formulate and prove a discrete version of the Hodge decomposition theorem.
We give an example illustrating how cohomology groups are calculated for our discrete model.
Note that our construction of discrete versions of the Hodge-Dirac and Laplace operators is very close to the construction in Section 5 of \cite{CKSV}. Matrix forms of these discrete operators on the torus are the same in both of case. The difference between our approach and one in \cite{CKSV} is in the definitions of discrete Hodge-Dirac and Laplace operators. In \cite{CKSV}, the definitions are given in terms of the exterior derivative and boundary operators. Meanwhile, as in the continual theory, we define these operators in terms of the exterior derivative and its adjoint.

 \section{Discrete model}
 In this section, we briefly review the construction of a discrete exterior calculus framework, which was initiated in \cite{Dezin} and developed in e.g., \cite{S2,S3}.

The starting point for consideration is a two-dimensional chain complex (a combinatorial model of $\mathbb{R}^2$).
 Let  the sets  $\{x_k\}$ and $\{e_k\}$, $k\in {\mathbb Z}$, be the generators of free abelian groups of
 zero-dimensional and one-dimensional
chains of the one-dimensional complex $C=C_0\oplus C_1$.
The free abelian group is understood as the direct sum of infinity cyclic groups generated by $\{x_k\}$, $\{e_{k}\}$.
The boundary operator $\partial: C_1\rightarrow C_0$ is the homomorphism defined by
$\partial e_k=x_{k+1}-x_k$
and the boundary of every zero chain is defined to be zero.
 Geometrically we can interpret the zero-dimensional basis elements $x_k$ as points of the real line and the one-dimensional basis elements $e_k$ as open intervals between points. We call the complex  $C$ a combinatorial real line.
Let the tensor product
\begin{equation*}
C(2)=C\otimes C
\end{equation*}
 be a combinatorial model of the two-dimensional Euclidean space
 $\mathbb{R}^2$.
The two-dimensional complex $C(2)=C_0(2)\oplus C_1(2)\oplus C_2(2)$ contains the free abelian groups of $r$-chains, $r=0,1,2$, generated by the basic elements
 \begin{align*}
x_{k,s}=x_k\otimes x_s,  \qquad
e_{k,s}^1=e_k\otimes x_s, \qquad e_{k,s}^2=x_k\otimes e_s, \qquad  V_{k,s}=e_k\otimes e_s,
\end{align*}
where $k,  s \in {\mathbb Z}$.
It is convenient to
introduce the shift operators  $\tau,\sigma$ in the set of indices by
\begin{equation}\label{2.1}
\tau k=k+1 \qquad \sigma k=k-1.
\end{equation}
 The boundary operator $\partial: C_r(2)\rightarrow C_{r-1}(2)$
 is given by
\begin{align}\label{2.2}
\partial x_{k,s}=0, \qquad  \partial e_{k,s}^1=x_{\tau k,s}-x_{k,s} \qquad  \partial e_{k,s}^2=x_{k, \tau s}-x_{k,s},\nonumber \\
\partial V_{k,s}=e_{k,s}^1+e_{\tau k,s}^2-e_{k, \tau s}^1-e_{k,s}^2.
\end{align}
The definition \eqref{2.2} is extended to arbitrary chains by linearity.

Let $K(2)$ be a complex of cochains with real coefficients. The cochain complex $K(2)$ with a coboundary operator defined in it is the dual object to the chain complex $C(2)$. It has a similar structure to $C(2)$ and consists of cochains of dimension  0, 1 and 2.
Then $K(2)$ can be expressed by
\begin{equation*}
K(2)=K^0(2)\oplus K^1(2)\oplus K^2(2),
\end{equation*}
where $K^r(2)$ is the set of all $r$-cochains. We will call cochains forms (or discrete forms) emphasizing their relationship with differential forms. Then the complex $K(2)$ is a discrete analogue of the grade algebra of differential forms  $\Lambda(\mathbb{R}^2)$.
Denote by $\{x^{k,s}\}$,  \
$\{e^{k,s}_1, \ e^{k,s}_2\}$ and $\{V^{k,s}\}$ the basis elements of $K^0(2)$, $K^1(2)$ and $K^2(2)$ respectively.
The pairing is defined with the basis elements of $C(2)$  by the rule
\begin{equation}\label{2.3}
\langle x_{k,s}, \ x^{i,j}\rangle=\langle e_{k,s}^1, \ e^{i,j}_1\rangle=
\langle e_{k,s}^2, \ e^{i,j}_2\rangle=\langle V_{k,s}, \ V^{i,j} \rangle=\delta_k^i\delta_s^j,
\end{equation}
where $\delta_k^i$ is the Kronecker delta.
The operation \eqref{2.3} is linearly extended to arbitrary chains and cochains.
     Let $\overset{r}{\omega}\in K^r(2)$, then we have
  \begin{equation}\label{2.4}
  \overset{0}{\omega}=\sum_{k,s}\overset{0}{\omega}_{k,s}x^{k,s}, \quad
  \overset{1}{\omega}=\sum_{k,s}(\omega^1_{k,s}e_1^{k,s}+\omega^2_{k,s}e_2^{k,s}), \quad \overset{2}{\omega}=\sum_{k,s}\overset{2}{\omega}_{k,s}V^{k,s},
\end{equation}
where  $\overset{0}{\omega}_{k,s}$, $\omega^1_{k,s}$,  $\omega^2_{k,s}$ and  $\overset{2}{\omega}_{k,s}$ are real numbers
for any $k, s \in {\mathbb Z}$.

The coboundary operator $d^c: K^r(2)\rightarrow K^{r+1}(2)$ is defined by
\begin{equation}\label{2.5}
\langle\partial a_{r+1}, \ \overset{r}{\omega}\rangle=\langle a_{r+1}, \ d^c\overset{r}{\omega}\rangle,
\end{equation}
where $a_{r+1}\in C_{r+1}(2)$.  The operator $d^c$ is an analog of the exterior differential.
From the above it follows that
\begin{equation*}\label{}
 d^c\overset{2}{\omega}=0 \quad \mbox{and} \quad d^cd^c\overset{r}{\omega}=0 \quad \mbox{for any} \quad r=0,1.
\end{equation*}
By \eqref{2.2} and \eqref{2.3} we can calculate
\begin{equation}\label{2.6}
d^c\overset{0}{\omega}=\sum_{k,s}(\Delta_k\overset{0}{\omega}_{k,s})e_1^{k,s}+(\Delta_s\overset{0}{\omega}_{k,s})e_2^{k,s},
\end{equation}
\begin{equation}\label{2.7}
d^c\overset{1}{\omega}=\sum_{k,s}(\Delta_k\omega_{k,s}^2-\Delta_s\omega_{k,s}^1)V^{k,s},
\end{equation}
where  $\Delta_k$ and $\Delta_s$ are the difference operators defined by
\begin{equation}\label{2.8}
\Delta_k\overset{r}{\omega}_{k,s}=\overset{r}{\omega}_{\tau k,s}-\overset{r}{\omega}_{k,s}, \qquad \Delta_s\overset{r}{\omega}_{k,s}=\overset{r}{\omega}_{k,\tau s}-\overset{r}{\omega}_{k,s}.
\end{equation}
Here  $\overset{r}{\omega}_{k,s}$ is a component of $\overset{r}{\omega}\in K^r(2)$  and $\tau$ is given by \eqref{2.1}.
 Note that $\overset{1}{\omega}_{k,s}=\{\omega^1_{k,s}, \ \omega^2_{k,s}\}$.

We now consider   a multiplication of discrete forms which is an analogue of the
exterior multiplication for differential forms.  Denote by  $\cup$ this multiplication.
For the basis elements of $K(2)$ the $\cup$-multiplication is defined as follows
\begin{equation*}\label{}
x^{k,s}\cup x^{k,s}=x^{k,s}, \qquad x^{k,s}\cup e^{k,s}_1=e^{k,s}_1, \qquad x^{k,s}\cup e^{k,s}_2=e^{k,s}_2,
\end{equation*}
\begin{equation*}\label{}
x^{k,s}\cup V^{k,s}=V^{k,s}, \qquad e^{k,s}_1\cup x^{\tau k,s}=e^{k,s}_1, \qquad e^{k,s}_2\cup x^{k, \tau s}=e^{k,s}_2,
\end{equation*}
\begin{equation*}\label{}
V^{k,s}\cup x^{\tau k,\tau s}=V^{k,s}, \qquad e^{k,s}_1\cup e^{\tau k,s}_2=V^{k,s}, \qquad e^{k,s}_2\cup e^{k, \tau s}_1=-V^{k,s},
\end{equation*}
supposing the product to be zero in all other cases. The operation is extended to arbitrary forms by linearity.
It is important to note that this definition leads to the following discrete counterpart of the Leibniz rule for differential forms.

\begin{prop}
Let $\overset{r}{\omega}\in K^r(2)$ and $\overset{p}{\varphi}\in K^p(2)$.
Then
\begin{equation}\label{2.9}
 d^c(\overset{r}{\omega}\cup\overset{p}{\varphi})=d^c\overset{r}{\omega}\cup\overset{p}{\varphi}+(-1)^r\overset{r}{\omega}\cup
d^c\overset{p}{\varphi}.
\end{equation}
\end{prop}
This  was proved by Dezin \cite{Dezin}.

Define the operation $\ast: K^r(2)\rightarrow  K^{2-r}(2)$  by the rule
\begin{equation}\label{2.10}
\ast x^{k,s}=V^{k,s}, \quad \ast e^{k,s}_1=e^{\tau k,s}_2, \quad \ast e^{k,s}_2=-e^{k,\tau s}_1, \quad \ast V^{k,s}=x^{\tau k, \tau s}.
\end{equation}
Again, the operation is extended to arbitrary forms by linearity. This operation is a discrete analogue of the Hodge star operator.
It is true that for any $\overset{r}{\omega}\in K^r(2)$  we have
  \begin{equation}\label{2.11}
  \overset{r}{\omega}\cup\ast\overset{r}{\omega}=\sum_{k,s}(\overset{r}{\omega}_{k,s})^2V^{k,s}.
\end{equation}
Consider the two-dimensional finite chain $V\in C_2(2)$ with unit coefficients of the form
\begin{equation}\label{2.12}
  V=\sum_{k,s}V_{k,s},  \qquad k=1,2, ..., N, \quad s=1,2,.., M.
\end{equation}
This chain imitates a rectangle.
Using \eqref{2.2} we have
\begin{equation}\label{2.13}
  \partial V=\sum_{k=1}^Ne^1_{k,1}+\sum_{s=1}^Me^2_{N,s}-\sum_{k=1}^Ne^1_{k,M}-\sum_{s=1}^Me^2_{1,s}.
\end{equation}
Then for forms
$\overset{r}{\varphi}, \ \overset{r}{\omega}\in K^r(2)$ of the same degree $r$ the inner
 product over the set  $V$ is defined  by the rule
 \begin{equation}\label{2.14}
 (\overset{r}{\varphi}, \ \overset{r}{\omega})_V=\langle V, \ \overset{r}{\varphi}\cup\ast\overset{r}{\omega}\rangle.
 \end{equation}
   For  forms of different degrees the product \eqref{2.14} is set equal to zero.
   From \eqref{2.11} and \eqref{2.3} we have
   \begin{equation*}\label{}
 (\overset{r}{\varphi}, \ \overset{r}{\omega})_V=\sum_{k=1}^N\sum_{s=1}^M\overset{r}{\varphi}_{k,s}\overset{r}{\omega}_{k,s}.
 \end{equation*}

 \begin{prop}
 Let $\overset{r}\varphi\in K^r(2)$  and $\overset{r+1}\omega\in K^{r+1}(2)$,  $r=0,1$. Then we have
\begin{equation}\label{2.15}
 (d^c\overset{r}\varphi, \ \overset{r+1}\omega)_V=\langle \partial V, \ \overset{r}{\varphi}\cup\ast\overset{r+1}{\omega}\rangle+(\overset{r}\varphi, \ \delta^c\overset{r+1}\omega)_V,
\end{equation}
 where
 \begin{equation}\label{2.16}
 \delta^c\overset{r+1}\omega=(-1)^{r+1}\ast^{-1}d^c\ast\overset{r+1}\omega
 \end{equation}
  is the
operator formally adjoint of $d^c$.
\end{prop}
Here $\ast^{-1}$ is the inverse of $\ast$, i.e., $\ast\ast^{-1}=1$. By \eqref{2.10} for the basic elements we have
\begin{equation*}\label{}
\ast^{-1} x^{k,s}=V^{\sigma k,\sigma s}, \quad \ast^{-1} e^{k,s}_1=-e^{k,\sigma s}_2, \quad \ast^{-1} e^{k,s}_2=e^{\sigma k, s}_1, \quad \ast^{-1} V^{k,s}=x^{k, s},
\end{equation*}
where $\sigma$ is given by \eqref{2.1}.
\begin{proof}
By Definitions \eqref{2.5}, \eqref{2.14} and Formula \eqref{2.9} we obtain
\begin{align*}\label{}
 (d^c\overset{r}\varphi, \ \overset{r+1}\omega)_V=\langle V, \ d^c\overset{r}{\varphi}\cup\ast\overset{r+1}{\omega}\rangle=
 \langle V, \ d^c(\overset{r}{\varphi}\cup\ast\overset{r+1}{\omega})-(-1)^r\overset{r}{\varphi}\cup d^c\ast\overset{r+1}{\omega}\rangle\\ =\langle\partial V, \ \overset{r}{\varphi}\cup\ast\overset{r+1}{\omega}\rangle+(-1)^{r+1}\langle V, \ \overset{r}{\varphi}\cup\ast\ast^{-1} d^c\ast\overset{r+1}{\omega}\rangle\\=\langle \partial V, \ \overset{r}{\varphi}\cup\ast\overset{r+1}{\omega}\rangle+\langle V, \ \overset{r}{\varphi}\cup\ast(\delta^c\overset{r+1}{\omega})\rangle.
 \end{align*}

\end{proof}
The operator $\delta^c: K^{r+1}(2) \rightarrow K^r(2)$ given by \eqref{2.16} is a discrete analogue of the codifferential $\delta$. For the 0-form $\overset{0}{\omega}\in K^0(2)$ we have $\delta^c\overset{0}{\omega}=0$.
It is obvious from \eqref{2.16} that
$\delta^c\delta^c\overset{r}{\omega}=0$ \ for any $r=1,2.$
 Using \eqref{2.6}--\eqref{2.8}, \eqref{2.10}  and \eqref{2.16} we can calculate
\begin{equation}\label{2.17}
\delta^c\overset{1}{\omega}=\sum_{k,s}(-\Delta_k\omega_{\sigma k,s}^{1}-\Delta_s\omega_{k,\sigma s}^{2})x^{k,s},
\end{equation}
\begin{equation}\label{2.18}
\delta^c\overset{2}{\omega}=\sum_{k,s}(\Delta_s\overset{2}{\omega}_{k,\sigma s})e_1^{k,s}
-(\Delta_k\overset{2}{\omega}_{\sigma k,s})e_2^{k,s}.
\end{equation}
In the particular case $r=0$, the equality \eqref{2.15} can be expressed as
\begin{align*}\label{}
 \sum_{k=1}^N\sum_{s=1}^M((\Delta_k\overset{0}{\varphi}_{k,s}){\omega}_{k,s}^1+(\Delta_s\overset{0}{\varphi}_{k,s}){\omega}_{k,s}^2)\\=
 \sum_{k=1}^N(\overset{0}{\varphi}_{k,\tau M}\omega^2_{k,M}-
\overset{0}{\varphi}_{k,1}\omega^2_{k,0})+\sum_{s=1}^M(\overset{0}{\varphi}_{\tau N,s}\omega^1_{N,s}-
\overset{0}{\varphi}_{1,s}\omega^1_{0,s})\\+
\sum_{k=1}^N\sum_{s=1}^M\overset{0}{\varphi}_{k,s}(-\Delta_k\omega_{\sigma k,s}^{1}-\Delta_s\omega_{k,\sigma s}^{2}),
 \end{align*}
 where $\overset{1}{\omega}_{k,s}=\{\omega^1_{k,s}, \ \omega^2_{k,s}\}$.
 The similar equality holds in the case $r=1$.

It should be noted that the relation \eqref{2.15} includes not only the forms with the components $\overset{r}{\varphi}_{k,s}$ and $\overset{r+1}{\omega}_{k,s}$, where the subscripts $k, s$ would run only over the values from \eqref{2.12}, but also the components
$\overset{r}{\varphi}_{0,s}$, \ $\overset{r}{\varphi}_{\tau N,s}$, \ $\overset{r+1}{\omega}_{0,s}$, $\overset{r+1}{\omega}_{\tau N,s}$, \
$\overset{r}{\varphi}_{k,0}$, \ $\overset{r}{\varphi}_{k, \tau M}$, $\overset{r+1}{\omega}_{k,0}$  and $\overset{r+1}{\omega}_{k, \tau M}$.

Let us set
\begin{align}\label{2.19}
\overset{r}{\omega}_{0,s}=\overset{r}{\omega}_{N,s}, \qquad \overset{r}{\omega}_{\tau N,s}=\overset{r}{\omega}_{1,s},  \qquad s=1,2,.., M, \nonumber \\
\overset{r}{\omega}_{k,0}=\overset{r}{\omega}_{k,M}, \qquad \overset{r}{\omega}_{k,\tau M}=\overset{r}{\omega}_{k,1}, \qquad k=1,2, ..., N.
\end{align}
For $r$-forms satisfying conditions \eqref{2.19} the inner product \eqref{2.14} generates the finite-dimensional Hilbert spaces $H^r(V)$.
Now we consider the operators
\begin{equation*}\label{}
d^c: H^r(V) \rightarrow H^{r+1}(V),  \qquad \delta^c: H^{r+1}(V) \rightarrow H^r(V).
\end{equation*}
\begin{prop}
 Let $\overset{r}\varphi\in H^r(V)$  and $\overset{r+1}\omega\in H^{r+1}(V)$,  $r=0,1$. Then we have
\begin{equation}\label{2.20}
 (d^c\overset{r}\varphi, \ \overset{r+1}\omega)_V=(\overset{r}\varphi, \ \delta^c\overset{r+1}\omega)_V.
\end{equation}
\end{prop}
\begin{proof}
In fact,  by use of conditions \eqref{2.19} one has
\begin{equation*}\label{}
\langle \partial V, \ \overset{0}{\varphi}\cup\ast\overset{1}{\omega}\rangle=\sum_{k=1}^N(\overset{0}{\varphi}_{k,\tau M}\omega^2_{k,M}-
\overset{0}{\varphi}_{k,1}\omega^2_{k,0})+\sum_{s=1}^M(\overset{0}{\varphi}_{\tau N,s}\omega^1_{N,s}-
\overset{0}{\varphi}_{1,s}\omega^1_{0,s})=0,
\end{equation*}
\begin{equation*}\label{}
\langle \partial V, \ \overset{1}{\varphi}\cup\ast\overset{2}{\omega}\rangle=\sum_{k=1}^N(\varphi^1_{k,1}\overset{2}{\omega}_{k,0}-
\varphi^1_{k, \tau M}\overset{2}{\omega}_{k,M})+\sum_{s=1}^M(\varphi^2_{\tau N,s}\overset{2}{\omega}_{N,s}-
\varphi^2_{1, s}\overset{2}{\omega}_{0,s})=0,
\end{equation*}
where $\overset{1}{\omega}_{k,s}=\{\omega^1_{k,s}, \ \omega^2_{k,s}\}$ and $\overset{1}{\varphi}_{k,s}=\{\varphi^1_{k,s}, \ \varphi^2_{k,s}\}$.
Hence by Proposition~2.2 it follows \eqref{2.20}.
\end{proof}

\section{Discrete Hodge decomposition}
In this section, we discuss the properties of discrete analogues of the Laplacian and Hodge-Dirac operators using the concepts of the previous section. We also present a discrete version of the Hodge decomposition theorem, emphasizing  that  it provides an exact counterpart to the continuum theory.

Let us consider the operator
\begin{equation*}\label{}
\Delta^c=d^c\delta^c+\delta^cd^c: H^r(V) \rightarrow H^r(V).
\end{equation*}
This is a discrete analogue of the Laplacian \eqref{1.2}.
\begin{prop}
 For any  $r$-form $\varphi\in H^r(V)$   we have
\begin{equation*}\label{}
 (\Delta^c\varphi, \ \varphi)_V\geq 0
\end{equation*}
and $(\Delta^c\varphi, \ \varphi)_V=0$ if and only if $\Delta^c\varphi=0$.
\end{prop}
\begin{proof}
By Proposition~2.3 one has
\begin{align*}\label{}
 (\Delta^c\varphi, \ \varphi)_V=(d^c\delta^c\varphi, \ \varphi)_V+(\delta^cd^c\varphi, \ \varphi)_V\\=
 (\delta^c\varphi, \ \delta^c\varphi)_V+(d^c\varphi, \ d^c\varphi)_V=\|\delta^c\varphi\|^2+\|d^c\varphi\|^2,
\end{align*}
where $\|\cdot\|$ denotes the norm and
\begin{equation*}
 \|\varphi\|^2=(\varphi, \ \varphi)_V=\sum_{k=1}^N\sum_{s=1}^M({\varphi}_{k,s})^2.
 \end{equation*}
From this if $(\Delta^c\varphi, \ \varphi)_V=0$ then $\|\delta^c\varphi\|^2=0$ and $\|d^c\varphi\|^2=0$.  It gives $\delta^c\varphi=0$ and $d^c\varphi=0$. Hence
\begin{equation*}
\Delta^c\varphi=d^c\delta^c\varphi+\delta^cd^c\varphi=0.
\end{equation*}
\end{proof}
\begin{cor}
$\Delta^c\varphi=0$ if and only if $\delta^c\varphi=0$ and $d^c\varphi=0$.
\end{cor}
\begin{prop}
 The operator  $\Delta^c: H^r(V) \rightarrow H^r(V)$  is self-adjoint, i.e.
\begin{equation*}\label{}
 (\Delta^c\varphi, \ \omega)_V=(\varphi, \ \Delta^c\omega)_V.
\end{equation*}
\end{prop}
\begin{proof}
By  \eqref{2.20} it is obvious.
\end{proof}
Consider the spaces
\begin{equation*}
R_{d^c}^r=\{d^c\varphi\in H^r(V): \ \varphi\in H^{r-1}(V)\},
\end{equation*}
\begin{equation*}
R_{\delta^c}^r=\{\delta^c\omega\in H^r(V): \  \omega\in H^{r+1}(V)\},
\end{equation*}
and
\begin{equation*}
N_{\Delta^c}^r=\{\psi\in H^r(V): \  \Delta^c\psi=0\}.
\end{equation*}
By analogy with the continuum case the discrete $r$-form $\omega$ is called closed if $d^c\omega=0$ and exact if $\omega\in R_{d^c}^r$.
\begin{prop}
For each $r=0,1,2$ we have the direct sum decomposition
 \begin{equation*}\label{}
 H^r(V)=R_{d^c}^r\oplus R_{\delta^c}^r\oplus N_{\Delta^c}^r.
\end{equation*}
\end{prop}
\begin{proof}
The space $H^1(V)$ decomposes into
\begin{equation}\label{3.1}
 H^1(V)=R_{d^c}^1\oplus  N_{\delta^c}^1,  \qquad
 H^1(V)=R_{\delta^c}^1\oplus N_{d^c}^1,
\end{equation}
where $N_{\delta^c}^1$ and $N_{d^c}^1$ are the orthogonal complements of the corresponding spaces.
For any $\omega\in R_{\delta^c}^1$ we have $\omega=\delta^c\psi$ and
\begin{equation*}\label{}
 (d^c\varphi, \ \omega)_V=(d^c\varphi, \ \delta^c\psi)_V=(d^cd^c\varphi, \ \psi)_V=0
\end{equation*}
for any $\varphi\in H^0(V)$. Therefore $\omega$ is orthogonal to $R_{d^c}^1$. It follows that
\begin{equation*}\label{}
 R_{\delta^c}^1\subset N_{\delta^c}^1.
\end{equation*}
Similarly we find
$R_{d^c}^1\subset N_{\delta^c}^1$.
Hence \eqref{3.1} becomes
\begin{equation*}\label{}
 H^1(V)=R_{d^c}^1\oplus  R_{\delta^c}^1\oplus N^1,
 \end{equation*}
 where
 \begin{equation*}\label{}
 N^1=N_{\delta^c}^1\cap N_{d^c}^1=\{\varphi\in H^1(V): \  d^c\varphi=0, \ \delta^c\varphi=0\}.
 \end{equation*}
 By Corollary~3.2 we have  $N^1=N_{\Delta^c}^1$.

 Similar reasonings apply  to the spaces $H^0(V)$ and $H^2(V)$. Thus we have
 \begin{equation*}\label{}
 H^0(V)=R_{\delta^c}^0\oplus N_{\Delta^c}^0,  \qquad
 H^2(V)=R_{d^c}^2\oplus  N_{\Delta^c}^2,
\end{equation*}
since $R_{d^c}^0=\{0\}$ and $R_{\delta^c}^2=\{0\}$.
\end{proof}
The Proposition~3.4 is a discrete version of the well-known Hodge decomposition theorem (see, e.g., \cite{Warner}).

Let $\Omega$ be an inhomogeneous discrete form, i.e. $\Omega=\overset{0}\omega+\overset{1}\omega+\overset{2}\omega$, where $\overset{r}\omega\in H^r(V)$. The inner product \eqref{2.14} can be extend to an inner product of inhomogeneous discrete forms by the rule
\begin{equation}\label{3.2}
 (\Omega, \ \Phi)_V=\sum_{r=0}^2(\overset{r}{\omega}, \ \overset{r}{\varphi})_V,
 \end{equation}
 where $\Phi=\overset{0}\varphi+\overset{1}\varphi+\overset{2}\varphi$.
 The inner product \eqref{3.2} generates the finite-dimensional Hilbert space $H(V)$. It is true that
 \begin{equation*}\label{}
 H(V)=H^0(V)\oplus  H^1(V)\oplus H^2(V).
 \end{equation*}
 By Proposition~3.4 the following holds
 \begin{equation}\label{3.3}
 H(V)=R_{d^c}\oplus R_{\delta^c}\oplus N_{\Delta^c},
 \end{equation}
 where
 \begin{equation*}\label{}
 R_{d^c}=R_{d^c}^1\oplus R_{d^c}^2, \qquad  R_{\delta^c}=R_{\delta^c}^0\oplus R_{\delta^c}^1,
 \end{equation*}
and
\begin{equation*}\label{}
  N_{\Delta^c}= N_{\Delta^c}^0\oplus  N_{\Delta^c}^1\oplus N_{\Delta^c}^2.
 \end{equation*}
 The discrete Hodge-Dirac operator is defined as
 \begin{equation}\label{3.4}
d^c+\delta^c: H(V) \rightarrow H(V).
\end{equation}
 \begin{prop}
 Operator  \eqref{3.4}  is self-adjoint with respect to the inner product \eqref{3.2}, i.e.
\begin{equation*}\label{}
 ((d^c+\delta^c)\Phi, \ \Omega)_V=(\Phi, \ (d^c+\delta^c)\Omega)_V.
\end{equation*}
\end{prop}
\begin{proof}
By  \eqref{2.20} it is obvious.
\end{proof}
\begin{prop}
 \begin{equation*}\label{}
(d^c+\delta^c)\Omega=0 \quad \mbox{if and only if} \quad \Omega \in N_{\Delta^c}.
\end{equation*}
\end{prop}
\begin{proof}
The equation $(d^c+\delta^c)\Omega=0$ can be written as
\begin{align*}\label{}
\delta^c\overset{1}{\omega}=0, \quad  d^c\overset{1}{\omega}=0, \\
d^c\overset{0}{\omega}+\delta^c\overset{2}{\omega}=0.
\end{align*}
Let $\Omega \in N_{\Delta^c}$. This means that $\overset{0}{\omega}\in N_{\Delta^c}^0$, $\overset{1}{\omega}\in N_{\Delta^c}^1$ and $\overset{2}{\omega}\in N_{\Delta^c}^2$
By Corollary~3.2, $\overset{1}{\omega}\in N_{\Delta^c}^1$ if and only if $\delta^c\overset{1}{\omega}=0$ and  $d^c\overset{1}{\omega}=0$.
 It is easy to show that
\begin{equation}\label{3.5}
\|d^c\overset{0}{\omega}+\delta^c\overset{2}{\omega}\|^2 =(d^c\overset{0}{\omega}+\delta^c\overset{2}{\omega}, \ d^c\overset{0}{\omega}+\delta^c\overset{2}{\omega})_V=\|d^c\overset{0}{\omega}\|^2+\|\delta^c\overset{2}{\omega}\|^2
\end{equation}
for any $\overset{0}{\omega}\in H^0(V)$ and $\overset{2}{\omega}\in H^2(V)$.
For $\overset{0}{\omega}\in N_{\Delta^c}^0$ and $\overset{2}{\omega}\in N_{\Delta^c}^2$,  by \eqref{2.20} and \eqref{3.5}, we have
\begin{equation*}\label{}
0=(\delta^cd^c\overset{0}{\omega}, \overset{0}{\omega})_V+(d^c\delta^c\overset{2}{\omega}, \overset{2}{\omega})_V=(d^c\overset{0}{\omega}, \ d^c\overset{0}{\omega})_V+(\delta^c\overset{2}{\omega}, \ \delta^c\overset{2}{\omega})_V=
\|d^c\overset{0}{\omega}+\delta^c\overset{2}{\omega}\|^2
\end{equation*}
and thus $d^c\overset{0}{\omega}+\delta^c\overset{2}{\omega}=0$. The converse is trivially true.
\end{proof}

\begin{prop}
For any inhomogeneous form $F\in R_{d^c}\oplus R_{\delta^c}$, there exists a unique form $\Omega\in R_{d^c}\oplus R_{\delta^c}$ which is a solution to the equation
 \begin{equation}\label{3.6}
(d^c+\delta^c)\Omega=F.
\end{equation}
\end{prop}
\begin{proof}
Since the operator \eqref{3.4} is self-adjoint and $H(V)$ is a finite-dimensional Hilbert space, the existence  of the solution is a consequence of the uniqueness of the solution and vice versa.
By Proposition~3.6, $(d^c+\delta^c)\Omega=0$ implies $\Omega\in N_{\Delta^c}$. From this and by \eqref{3.3}, the uniqueness of the solution $\Omega\in R_{d^c}\oplus R_{\delta^c}$ of Eq.~\eqref{3.6} follows immediately.
\end{proof}
\begin{cor}
If $\Omega\in H(V)$  is a solution of Eq.~\eqref{3.6} then the following holds
 \begin{equation*}\label{}
\|\Omega\|^2\leq c(\|d^c\Omega\|^2+\|\delta^c\Omega\|^2)+\|\Omega_{\Delta^c}\|^2,
\end{equation*}
where $c$ is a constant, $d^c\Omega=d^c\overset{0}{\omega}+d^c\overset{1}{\omega}$,  \ $\delta^c\Omega=\delta^c\overset{1}{\omega}+\delta^c\overset{2}{\omega}$
and $\Omega_{\Delta^c}$ is the projection of $\Omega$ onto  $N_{\Delta^c}$.
\end{cor}

\section{Combinatorial torus}

In this section, we consider an example of a discrete model of the torus in detail. We recall that the torus can be regarded as the topological space obtained by taking a rectangle  and identifying each pair of opposite sides with the same orientation.
Let consider the partitioning of the plane $\mathbb{R}^2$ by the straight lines $x=k$ and $y=s$, where $k,s\in\mathbb{Z}$.
Denote by $V_{k,s}$ an open square bounded by the lines $x=k, \ x=\tau k, \ y=s$ and $y=\tau s$, where $\tau$ is given by \eqref{2.1}.
Denote the vertices of $V_{k,s}$ by $x_{k,s}, \ x_{\tau k,s}, \ x_{k,\tau s}$, \ $x_{\tau k, \tau s}$. Let $e_{k,s}^1$ and $e_{k,s}^2$ be the open intervals $(x_{k,s}, \ x_{\tau k,s})$ and $(x_{k,s}, \ x_{k, \tau s})$, respectively.
Such introduced geometric objects can be identified with the combinatorial objects we have considered in the previous sections.
We identify the collection $V_{k,s}$ with $V$ given by \eqref{2.12} and let $N=M=2$. In this case the conditions \eqref{2.19} take the form
\begin{equation}\label{4.1}
\overset{r}{\omega}_{0,s}=\overset{r}{\omega}_{2,s}, \qquad \overset{r}{\omega}_{3,s}=\overset{r}{\omega}_{1,s}, \qquad
\overset{r}{\omega}_{k,0}=\overset{r}{\omega}_{k,2}, \qquad \overset{r}{\omega}_{k,3}=\overset{r}{\omega}_{k,1},
\end{equation}
where $k=1,2$ and $s=1,2$.
If we identify the points and the intervals on the boundary of $V$ in the following way
\begin{align}\label{4.2}
x_{1,1}=x_{3,1}=x_{1,3}=x_{3,3}, \qquad x_{1,2}=x_{3,2}, \qquad x_{2,1}=x_{2,3}, \nonumber \\
e_{1,1}^1=e_{1,3}^1, \qquad e_{2,1}^1=e_{2,3}^1, \qquad e_{1,1}^2=e_{3,1}^2, \qquad e_{1,2}^2=e_{3,2}^2,
\end{align}
we obtain the geometric object which is homomorphic to the torus (see Figure~1).
Denote by $C(T)$ the complex $C(2)$ which corresponds to the introduced geometric object. We call $C(T)$  a combinatorial torus.
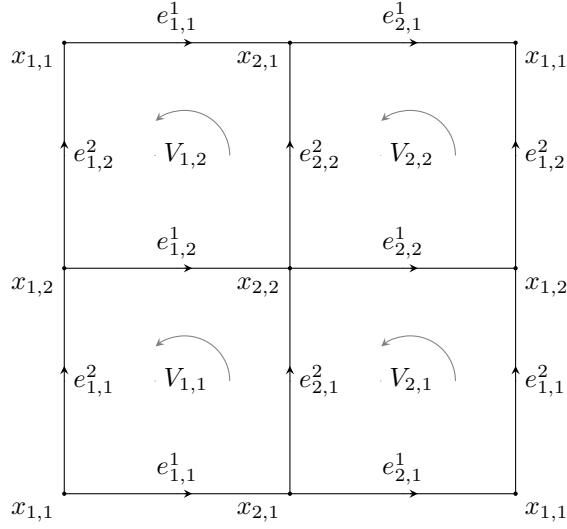
\begin{figure}[h]
\begin{center}
\begin{tikzpicture}
\draw[black, thin] (0,0) -- (3,0);
\draw[black, thin, ->, >=stealth] (0,1.7);
\draw[black, thin, ->, >=stealth] (3,1.7);
\draw[black, thin, ->, >=stealth] (6,1.7);
\draw[black, thin, ->, >=stealth] (0,4.7);
\draw[black, thin, ->, >=stealth] (3,4.7);
\draw[black, thin, ->, >=stealth] (6,4.7);
\draw[black, thin, ->, >=stealth] (1.6,0) -- (1.7,0);
\draw[black, thin, ->, >=stealth] (4.6,0) -- (4.7,0);
\draw[black, thin, ->, >=stealth] (1.6,3) -- (1.7,3);
\draw[black, thin, ->, >=stealth] (4.6,3) -- (4.7,3);
\draw[black, thin, ->, >=stealth] (1.6,6) -- (1.7,6);
\draw[black, thin, ->,>=stealth] (4.6,6) -- (4.7,6);
\draw[black, thin] (3,6) -- (6,6);
\draw[black, thin] (0,6) -- (3,6);
\draw[black, thin] (0,0) -- (0,3);
\draw[black, thin] (0,3) -- (0,6);
\draw[black, thin] (3,0) -- (3,3);
\draw[black, thin] (3,0) -- (6,0);
\draw[black, thin] (3,3) -- (3,6);
\draw[black, thin] (3,3) -- (6,3);
\draw[black, thin] (0,3) -- (3,3);
\draw[black, thin] (6,0) -- (6,3);
\draw[black, thin] (6,3) -- (6,6);
\draw[gray,thin, ->, >=stealth] (2.2,1.5) arc (0:130:0.6);
\draw[gray,thin, ->, >=stealth] (5.2,1.5) arc (0:130:0.6);
\draw[gray,thin, ->, >=stealth] (2.2,4.5) arc (0:130:0.6);
\draw[gray,thin, ->, >=stealth] (5.2,4.5) arc (0:130:0.6);
\filldraw[black] (0,0) circle (0.6pt) node[anchor=north east] {$x_{1,1}$};
\filldraw[black] (3,0) circle (0.6pt) node[anchor=north east] {$x_{2,1}$};
\filldraw[black] (0,6) circle (0.6pt) node[anchor=north east] {$x_{1,1}$};
\filldraw[black] (3,3) circle (0.6pt) node[anchor=north east] {$x_{2,2}$};
\filldraw[black] (1.5,0) circle (0pt) node[anchor=south] {$e^1_{1,1}$};
\filldraw[black] (1.5,6) circle (0pt) node[anchor=south] {$e^1_{1,1}$};
\filldraw[black] (4.5,0) circle (0pt) node[anchor=south] {$e^1_{2,1}$};
\filldraw[black] (4.5,6) circle (0pt) node[anchor=south] {$e^1_{2,1}$};

\filldraw[black] (0,1.5) circle (0pt) node[anchor=west] {$e^2_{1,1}$};
\filldraw[black] (0,4.5) circle (0pt) node[anchor=west] {$e^2_{1,2}$};
\filldraw[black] (3,4.5) circle (0pt) node[anchor=west] {$e^2_{2,2}$};
\filldraw[black] (6,4.5) circle (0pt) node[anchor=west] {$e^2_{1,2}$};
\filldraw[black] (6,1.5) circle (0pt) node[anchor=west] {$e^2_{1,1}$};
\filldraw[black] (3,1.5) circle (0pt) node[anchor=west] {$e^2_{2,1}$};
\filldraw[black] (1.5,3) circle (0pt) node[anchor=south] {$e^1_{1,2}$};
\filldraw[black] (4.5,3) circle (0pt) node[anchor=south] {$e^1_{2,2}$};
\filldraw[black] (0,3) circle (0.6pt) node[anchor=north east] {$x_{1,2}$};
\filldraw[black] (6,0) circle (0.6pt) node[anchor=north west] {$x_{1,1}$};
\filldraw[black] (6,3) circle (0.6pt) node[anchor=north west] {$x_{1,2}$};
\filldraw[black] (6,6) circle (0.6pt) node[anchor=north west] {$x_{1,1}$};
\filldraw[black] (3,6) circle (0.6pt) node[anchor=north east] {$x_{2,1}$};
\filldraw[black] (1.2, 1.5) circle (0pt) node[anchor=west] {$V_{1,1}$};
\filldraw[black] (4.2, 1.5) circle (0pt) node[anchor=west] {$V_{2,1}$};
\filldraw[black] (4.2, 4.5) circle (0pt) node[anchor=west] {$V_{2,2}$};
\filldraw[black] (1.2, 4.5) circle (0pt) node[anchor=west] {$V_{1,2}$};
\end{tikzpicture}
\caption{Combinatorial torus}
\end{center}
\end{figure}
It is clear that by \eqref{4.2} the conditions \eqref{4.1} hold for any $r$-form on the combinatorial torus.
The Hilbert space $H(V)$ considered in the previous sections can be regarded as a space of cochains of the complex $K(T)$ dual to $C(T)$.

Let now consider the forms $\varphi\in K^0(T)$, $\omega\in K^1(T)$ and $\psi\in K^2(T)$, that is
\begin{equation*}\label{}
  \varphi=\sum_{k=1}^2\sum_{s=1}^2\varphi_{k,s}x^{k,s}, \quad
  \omega=\sum_{k=1}^2\sum_{s=1}^2(u_{k,s}e_1^{k,s}+v_{k,s}e_2^{k,s}), \quad \psi=\sum_{k=1}^2\sum_{s=1}^2\psi_{k,s}V^{k,s}.
\end{equation*}
By  \eqref{4.2} for these forms the formulas \eqref{2.6}, \eqref{2.7}, \eqref{2.17} and \eqref{2.18} become
\begin{align}\label{4.3}
 d^c\varphi=(\varphi_{2,1}-\varphi_{1,1})e_1^{1,1}+(\varphi_{1,1}-\varphi_{2,1})e_1^{2,1}+(\varphi_{2,2}-\varphi_{1,2})e_1^{1,2}\nonumber\\+(\varphi_{1,2}-\varphi_{2,2})e_1^{2,2}+
  (\varphi_{1,2}-\varphi_{1,1})e_2^{1,1}+(\varphi_{1,1}-\varphi_{1,2})e_2^{1,2}\nonumber\\+(\varphi_{2,2}-\varphi_{2,1})e_2^{2,1}+(\varphi_{2,1}-\varphi_{2,2})e_2^{2,2},
\end{align}
\begin{align}\label{4.4}
 d^c\omega=(u_{1,1}-u_{1,2}+v_{2,1}-v_{1,1})V^{1,1}+(u_{2,1}-u_{2,2}-v_{2,1}+v_{1,1})V^{2,1}\nonumber\\
  +(u_{1,2}-u_{1,1}+v_{2,2}-v_{1,2})V^{1,2}+(u_{2,2}-u_{2,1}+v_{1,2}-v_{2,2})V^{2,2},
\end{align}
\begin{align}\label{4.5}
 \delta^c\omega=(u_{2,1}-u_{1,1}+v_{1,2}-v_{1,1})x^{1,1}+(u_{1,1}-u_{2,1}-v_{2,1}+v_{2,2})x^{2,1}\nonumber\\
  +(u_{2,2}-u_{1,2}+v_{1,1}-v_{1,2})x^{1,2}+(u_{1,2}-u_{2,2}+v_{2,1}-v_{2,2})x^{2,2},
\end{align}
\begin{align}\label{4.6}
 \delta^c\psi=(\psi_{1,1}-\psi_{1,2})e_1^{1,1}+(\psi_{2,1}-\psi_{2,2})e_1^{2,1}+(\psi_{2,2}-\psi_{1,2})e_2^{1,2}\nonumber\\+(\psi_{2,1}-\psi_{1,1})e_2^{1,1}+
  (\psi_{1,2}-\psi_{1,1})e_1^{1,2}+(\psi_{2,2}-\psi_{2,1})e_1^{2,2}\nonumber\\+(\psi_{1,2}-\psi_{2,2})e_2^{2,2}+(\psi_{1,1}-\psi_{2,1})e_2^{2,1}.
\end{align}
It should be noted that the formulas above can be expressed  in matric form.
Let introduce the following row vectors
\begin{equation*}\label{}
 [\varphi]=[\varphi_{1,1} \ \varphi_{2,1} \ \varphi_{1,2} \ \varphi_{2,2} ], \quad
 [\omega]=[u_{1,1} \ u_{2,1} \ v_{1,2} \ v_{1,1} \ u_{1,2} \ u_{2,2} \ v_{2,2} \ v_{2,1}],
 \end{equation*}
 \begin{equation*}\label{}
 [x]=[x^{1,1} \ x^{2,1} \ x^{1,2} \ x^{2,2} ], \quad
 [e]=[e_1^{1,1} \ e_1^{2,1} \ e_2^{1,2} \ e_2^{1,1} \ e_1^{1,2} \ e_1^{2,2} \ e_2^{2,2} \ e_2^{2,1}],
 \end{equation*}
 \begin{equation*}\label{}
 [\psi]=[\psi_{1,2} \ \psi_{2,2} \ \psi_{1,1} \ \psi_{2,1} ], \qquad
 [V]=[V^{1,2} \ V^{2,2} \ V^{1,1} \ V^{2,1}].
 \end{equation*}
 Denote by $[\cdot]^T$ a corresponding column vector.
 Then we have
 \begin{equation*}\label{}
 d^c\varphi=A[\varphi]^T[e], \quad d^c\omega=B[\omega]^T[V],
  \quad \delta^c\omega=A^T[\omega]^T[x], \quad \delta^c\psi=B^T[\psi]^T[e],
 \end{equation*}
 where
\begin{equation*}
A=
\begin{bmatrix}
-1 & 1 & 0 & 0 \\
1 & -1 & 0 & 0 \\
1 & 0 & -1 & 0 \\
-1 & 0 & 1 & 0 \\
0 & 0 & -1 & 1 \\
0 & 0 & 1 & -1 \\
0 & 1 & 0 & -1 \\
0 & -1 & 0 & 1
\end{bmatrix}, \quad
B=
\begin{bmatrix}
-1 & 0 & -1 & 0 & 1 & 0 & 1 & 0 \\
0 & -1 & 1 & 0 & 0 & 1 & -1 & 0\\
1 & 0 & 0 & -1 & -1 & 0 & 0 & 1\\
0 & 1 & 0 & 1 & 0 & -1 & 0 & -1
\end{bmatrix}
\end{equation*}
and $A^T$, $B^T$ are the transpose of $A$, $B$.

 The discrete Hodge-Dirac operator \eqref{3.4} on the combinatorial torus can be represented by the following block matrix
\begin{equation*}
\begin{bmatrix}
\bigzero & A^T & \bigzero \\
A & \bigzero & B^T \\
\bigzero & B & \bigzero
\end{bmatrix},
\end{equation*}
where $\bigzero$ is a zero square  matrix of the corresponding size.

In the same way the discrete Laplacian on the combinatorial torus can be written as
\begin{equation*}\label{}
 \Delta^c\varphi=D[\varphi]^T[x], \quad \Delta^c\omega=D_1[\omega]^T[e],
   \quad \Delta^c\psi=D[\psi]^T[V],
 \end{equation*}
 where
\begin{equation*}
D=
\begin{bmatrix}
4 & -2 & -2 & 0 \\
-2 & 4 & 0 & -2 \\
-2 & 0 & 4 & -2 \\
0 & -2 & -2 & 4
\end{bmatrix}, \
D_1=
\begin{bmatrix}
4 & -2 & 0 & 0 & -2 & 0 & 0 & 0 \\
-2 & 4 & 0 & 0 & 0 & -2 & 0 & 0\\
0 & 0 & 4 & -2 & 0 & 0 & -2 & 0\\
0 & 0 & -2 & 4 & 0 & 0 & 0 & -2 \\
-2 & 0 & 0 & 0 & 4 & -2 & 0 & 0 \\
0 & -2 & 0 & 0 & -2 & 4 & 0 & 0 \\
0 & 0 & -2 & 0 & 0 & 0 & 4 & -2 \\
0 & 0 & 0 & -2 & 0 & 0 & -2 & 4
\end{bmatrix}.
\end{equation*}

Let us define  analogues of the cohomology groups $\mathcal{H}^r(T)$ of the combinatorial torus $C(T)$.
The quotient space of the linear space of  closed $r$-forms
\begin{equation*}
N_{d^c}^r(T)=\{\omega\in K^r(T): \  d^c\omega=0\}
\end{equation*}
modulo the subspace of exact $r$-forms
\begin{equation*}
R_{d^c}^r(T)=\{\omega\in K^r(T): \ \exists\varphi\in K^{r-1}(T) \quad \omega=d^c\varphi\}
\end{equation*}
is called  the $r$-th cohomology group of $C(T)$, that is,
\begin{equation*}
\mathcal{H}^r(T)=N_{d^c}^r(T)/R_{d^c}^r(T).
\end{equation*}
Two closed $r$-forms  $\omega^1$ and $\omega^2$ are cohomologous, $\omega^1\sim\omega^2$, if and only if they differ by an exact form, i.e.,
\begin{equation*}
\omega^1\sim\omega^2 \Leftrightarrow \omega^1-\omega^2\in R_{d^c}^r(T).
\end{equation*}
An element of $\mathcal{H}^r(T)$ is thus an equivalence class $[\omega]$ of closed $r$-forms $\omega+d^c\varphi$, defined by the equivalence relation $\sim$. These equivalence classes endow $\mathcal{H}^r(T)$ with a group structure.

Calculation of $\mathcal{H}^0(T)$. Since there are no $-1$-forms, a $0$-form $\varphi$ can never be exact, i.e. $R_{d^c}^0(T)=\{0\}$.
If $\varphi\in N_{d^c}^0(T)$ then $d^c\varphi=0$. By \eqref{4.3} it follows immediately that  $\varphi_{1,1}=\varphi_{2,1}=\varphi_{1,2}=\varphi_{2,2}$. Hence
\begin{equation*}
\varphi=c(x^{1,1}+x^{2,1}+x^{1,2}+x^{2,2}),
\end{equation*}
 where $c\in\mathbb{R}$ is a constant. Thus $\mathcal{H}^0(T)$ is isomorphic to the group generated by  one independent generator which is isomorphic to $\mathbb{R}$  and we write
$\mathcal{H}^0(T)\cong\mathbb{R}$.

Calculation of $\mathcal{H}^1(T)$. Let $\omega^1=\{u^1, v^1\}\in R_{d^c}^1(T)$. Then $\omega^1=d^c \varphi$ for some  $\varphi\in K^0(T)$. From \eqref{4.3} we have
\begin{align*}\label{}
 u^1_{1,1}=\varphi_{2,1}-\varphi_{1,1}, \qquad u^1_{1,2}=\varphi_{2,2}-\varphi_{1,2}, \\
  u^1_{2,1}=\varphi_{1,1}-\varphi_{2,1}, \qquad u^1_{2,2}=\varphi_{1,2}-\varphi_{2,2}, \\
  v^1_{1,1}=\varphi_{1,2}-\varphi_{1,1}, \qquad v^1_{2,1}=\varphi_{2,2}-\varphi_{2,1}, \\ v^1_{1,2}=\varphi_{1,1}-\varphi_{1,2},\qquad v^1_{2,2}=\varphi_{2,1}-\varphi_{2,2}.
\end{align*}
It follows that
\begin{equation*}
u^1_{2,1}=-u^1_{1,1}, \qquad u^1_{1,2}=-u^1_{2,2}, \qquad v^1_{1,2}=-v^1_{1,1}, \qquad v^1_{2,1}=-v^1_{2,2}.
\end{equation*}
Hence any form $\omega^1\in R_{d^c}^1(T)$ can be written as
\begin{align*}\label{}
\omega^1=u^1_{1,1}e_1^{1,1}-u^1_{1,1}e_1^{2,1}-u^1_{2,2}e_1^{1,2}+u^1_{2,2}e_1^{2,2}\\+
  v^1_{1,1}e_2^{1,1}-v^1_{1,1}e_2^{1,2}-v^1_{2,2}e_2^{2,1}+v^1_{2,2}e_2^{2,2}\\
  =u^1_{1,1}(e_1^{1,1}-e_1^{2,1})+u^1_{2,2}(e_1^{2,2}-e_1^{1,2})+
  v^1_{1,1}(e_2^{1,1}-e_2^{1,2})+v^1_{2,2}(e_2^{2,2}-e_2^{2,1}).
\end{align*}
Let now $\omega=\{u, v\}$ be a closed 1-form, i. e., $\omega\in N_{d^c}^1(T)$ and $d^c \omega=0$.
By \eqref{4.4} the requirement $d^c \omega=0$ means  that
\begin{align*}\label{}
u_{1,1}-u_{1,2}+v_{2,1}-v_{1,1}=0,\\
u_{2,1}-u_{2,2}-v_{2,1}+v_{1,1}=0, \\
  u_{1,2}-u_{1,1}+v_{2,2}-v_{1,2}=0,\\
  u_{2,2}-u_{2,1}+v_{1,2}-v_{2,2}=0.
\end{align*}
From this we obtain
\begin{align*}\label{}
u_{1,1}-u_{1,2}+u_{2,1}-u_{2,2}=0,\\
v_{1,1}-v_{2,1}+v_{1,2}-v_{2,2}=0.
\end{align*}
Hence  any form $\omega\in N_{d^c}^1(T)$ can be written as
\begin{align*}\label{}
\omega=u_{1,1}e_1^{1,1}+(u_{1,2}-u_{1,1}+u_{2,2})e_1^{2,1}+u_{1,2}e_1^{1,2}+u_{2,2}e_1^{2,2} \\+
  v_{1,1}e_2^{1,1}+(v_{1,1}+v_{1,2}-v_{2,2})e_2^{2,1}+v_{1,2}e_2^{1,2}+v_{2,2}e_2^{2,2}.
 \end{align*}
 This yields
 \begin{align*}\label{}
\omega=u_{1,1}(e_1^{1,1}-e_1^{2,1})+u_{1,2}(e_1^{2,1}+e_1^{1,2})+u_{2,2}(e_1^{2,2}+e_1^{2,1}) \\+
  v_{1,1}(e_2^{1,1}+e_2^{2,1})+v_{1,2}(e_2^{2,1}+e_2^{1,2})+v_{2,2}(e_2^{2,2}-e_2^{2,1})\\
  =(u_{1,2}+u_{2,2})(e_1^{2,1}+e_1^{1,2})+(v_{1,2}+v_{1,1})(e_2^{2,1}+e_2^{1,2})+\omega^0,
  \end{align*}
  where
 \begin{equation*}
 \omega^0= u_{1,1}(e_1^{1,1}-e_1^{2,1})+u_{2,2}(e_1^{2,2}-e_1^{1,2})+
  v_{1,1}(e_2^{1,1}-e_2^{1,2})+v_{2,2}(e_2^{2,2}-e_2^{2,1})
  \end{equation*}
  and note that $\omega^0\in R_{d^c}^1(T)$. Hence
  \begin{equation*}
\omega\sim (u_{1,2}+u_{2,2})(e_1^{2,1}+e_1^{1,2})+(v_{1,2}+v_{1,1})(e_2^{2,1}+e_2^{1,2})\in\mathcal{H}^1(T).
  \end{equation*}
 This means that $\mathcal{H}^1(T)$ has two independent generators. Thus $\mathcal{H}^1(T)\cong\mathbb{R}^2$.

 Calculation of $\mathcal{H}^2(T)$. A 2-form
 \begin{equation*}
 \psi^1=\psi^1_{1,1}V^{1,1}+\psi^1_{2,1}V^{2,1}+\psi^1_{1,2}V^{1,2}+\psi^1_{2,2}V^{2,2}
 \end{equation*}
is an element of $R_{d^c}^2(T)$ if $\psi^1=d^c\omega$ for some 1-form $\omega=\{u,v\}$. Using \eqref{4.4} $\psi^1=d^c\omega$ gives rise to
\begin{align*}\label{}
\psi^1_{1,1}=u_{1,1}-u_{1,2}+v_{2,1}-v_{1,1},  \quad \psi^1_{2,1}=u_{2,1}-u_{2,2}-v_{2,1}+v_{1,1}, \\
 \psi^1_{1,2}=u_{1,2}-u_{1,1}+v_{2,2}-v_{1,2}, \quad \psi^1_{2,2}=u_{2,2}-u_{2,1}+v_{1,2}-v_{2,2}.
\end{align*}
Adding these equations we obtain
\begin{equation*}
 \psi^1_{1,1}+\psi^1_{2,1}+\psi^1_{1,2}+\psi^1_{2,2}=0.
 \end{equation*}
 Hence any element $\psi^1\in R_{d^c}^2(T)$ can be written as
 \begin{equation*}
 \psi^1=\psi^1_{1,1}V^{1,1}+\psi^1_{2,1}V^{2,1}+\psi^1_{1,2}V^{1,2}+(-\psi^1_{1,1}-\psi^1_{2,1}-\psi^1_{1,2})V^{2,2}.
 \end{equation*}
 Since $d^c\psi=0$ for any 2-form $\psi$, $N_{d^c}^2(T)=K^2(T)$. Any element of $N_{d^c}^2(T)$ can be expressed as
 \begin{equation*}
 \psi=(\psi_{1,1}+\psi_{2,1}+\psi_{1,2}+\psi_{2,2})V^{2,2}+\psi^0,
 \end{equation*}
 where
 \begin{equation*}
 \psi^0=\psi_{1,1}V^{1,1}+\psi_{2,1}V^{2,1}+\psi_{1,2}V^{1,2}+(-\psi_{1,1}-\psi_{2,1}-\psi_{1,2})V^{2,2}.
 \end{equation*}
 Thus
 \begin{equation*}
 \psi\sim(\psi_{1,1}+\psi_{2,1}+\psi_{1,2}+\psi_{2,2})V^{2,2},
 \end{equation*}
 since $\psi^0\in R_{d^c}^2(T)$.  This means that $\mathcal{H}^2(T)$ has only one independent generator. So that $\mathcal{H}^2(T)\cong\mathbb{R}$.

 Thus the cohomology groups are exactly the same as in the continuum case.

 \end{document}